\documentclass[11pt, a4paper]{article}

\pdfoutput=1

\usepackage{latexsym}
\usepackage[centertags]{amsmath}
\usepackage{amsfonts}
\usepackage{amssymb}
\usepackage{amsthm}
\usepackage{newlfont}
\usepackage{graphicx}
\usepackage{float}
\usepackage{hhline}
\usepackage{algorithm}
\usepackage{algpseudocode}
\usepackage{fullpage}
\usepackage{amsmath}
\usepackage{mdframed}
\usepackage{verbatim}

\def\denseformat{
\setlength{\textheight}{9.65in} \setlength{\textwidth}{6.8in}
\setlength{\oddsidemargin}{-0.3in} \setlength{\headsep}{10pt}
\setlength{\topmargin}{-0.33in} \setlength{\columnsep}{0.375in}
}
\denseformat
\algdef{SE}[RECEIVING]{Receiving}{EndReceiving}[1]{\textbf{upon
receiving}\ #1\ \algorithmicdo}{\algorithmicend\ \textbf{}}%
\algtext*{EndEvent}

\algdef{SE}[UPON]{Upon}{EndUpon}[1]{\textbf{upon
}\ #1\ \algorithmicdo}{\algorithmicend\ \textbf{}}%
\algtext*{EndEvent}

\restylefloat{table}
\restylefloat{verbatim}

\newmdenv[linewidth=0pt]{myframe}

\newcommand{\qedlemma}[1]{$\square$\;\small
\textit{Lemma~\ref{#1}}}
\newcommand{\qedthm}[1]{$\blacksquare$\;\small
\textit{Theorem~\ref{#1}}}

\newcommand{\defeq}{%
  \mathrel{\vbox{\offinterlineskip\ialign{%
    \hfil##\hfil\cr
    $\scriptscriptstyle\triangle$\cr
    $=$\cr
}}}}

\theoremstyle{definition}
\newtheorem{definition}{Definition}

\theoremstyle{plain}

\newtheorem{observation}{Observation}
\newtheorem{corollary}{Corollary}
\newtheorem{theorem}{Theorem}  %[section]
\newtheorem{lemma}{Lemma}[theorem]

\theoremstyle{remark}

\begin{document}

\title{ On Liveness of Dynamic Storage}

\begin{comment}
\author{Alexander Spiegelman
\vspace{1mm} \\
Idit Keidar\\
\vspace{3mm} \\
Technion - Israel Institute of Technology\\
Department of Electrical Engineering}
\date{}
\maketitle
\end{comment}

\author{%
  {Alexander Spiegelman}\\EE Department\\Technion, Haifa,
  Israel\\sashas@tx.technion.ac.il\\+972547553558
  \and
  {Idit Keidar}\\EE Department\\Technion, Haifa,
  Israel\\idish@ee.technion.ac.il\\
}
\date{}
\maketitle

\begin{abstract}

Dynamic distributed storage algorithms such as DynaStore,
Reconfigurable Paxos, RAMBO, and RDS, do not ensure liveness
(wait-freedom) in asynchronous runs with infinitely many
reconfigurations. We prove that this is inherent for
asynchronous dynamic storage algorithms, including ones that use
$\Omega$ or $\diamond S$ oracles. Our result holds even if only
\emph{one} process may fail, provided that machines that were
successfully removed from the system's configuration may be
switched off by an administrator. Intuitively, the
impossibility relies on the fact that a correct process can be
\emph{suspected} to have failed at any time, i.e., its failure
is indistinguishable to other processes from slow delivery of
its messages, and so the system should be able to reconfigure
without waiting for this process to complete its pending
operations.

To circumvent this result, we define a \emph{dynamic eventually
perfect failure detector}, and present an algorithm that uses it
to emulate wait-free dynamic atomic storage (with no
restrictions on reconfiguration rate). Together, our results
thus draw a sharp line between oracles like $\Omega$ and
$\diamond S$, which allow some correct process to continue to be
suspected forever, and  a dynamic eventually perfect one, which
does not.

\end{abstract}

\setcounter{page}{0}
\thispagestyle{empty}

\newpage

\section{Introduction}

Many works in the last decade have dealt with the
emulation of reliable storage via message passing in
dynamic
systems~\cite{ShraerDynadisc,dynastore,gilbert2003rambo,gilbert2010rambo,RDS2009reconfigurable,baldoni2009implementing,SMRlamport2010reconfiguring,SMRvirtually}.
The motivation behind such systems is to allow the current
configuration of participating processes to be changed;
once a process is removed from the current configuration,
a system administrator may shut it down, and the storage 
algorithm can no longer rely on it in order to ensure progress.
We elaborate more on the requirement from dynamic storage in
Section \ref{sec:models}.
To the best of our knowledge, no
previous dynamic storage solution ensures
completion of all operations in asynchronous runs with
unrestricted reconfigurations, as detailed in Section
\ref{sec:relatedWork}.

In Section \ref{sec:storageImpossibility}, we show that
this limitation is inherent as long as even \emph{one}
process that was not removed from the current configuration can
fail.
Specifically, we show that even a \emph{safe} register
emulation cannot guarantee liveness for all operations (i.e.,
wait-freedom) in asynchronous runs with unrestricted
reconfigurations.
The key to the impossibility proof is that, in asynchronous
models, a slow correct process whose messages are delayed can be
\emph{suspected} to be faulty by all other processes,
i.e., the two scenarios are indistinguishable and so the system
should be able to reconfigure without allowing it to complete
its operations.
Our result holds even if processes are
equipped with oracles like $\Omega$ or $\diamond S$, which allow
them to chose a leader and solve consensus in every
configuration, yet some correct process can continue to be
suspected forever.

On the other hand, with a failure detector that guarantees a
time after which correct processes are no longer suspected,
our proof does not hold.
Indeed, in Section \ref{sec:diamondP}, we define a dynamic
version of the eventually perfect failure detector $\diamond
P$~\cite{chandra1996unreliable}, which we call $\diamond P^D$,
and use it to implement a helping mechanism in order to
circumvent the impossibility result. We present an algorithm,
based on state machine replication, that emulates a wait-free
atomic dynamic multi-writer, multi-reader (MWMR) register, and
ensures liveness with unrestricted reconfigurations.
Unlike $\Omega$-based reconfigurable state machine
replication~\cite{SMRlamport2010reconfiguring,SMRvirtually}, our
implementation ensures completion of all operations.

Together, our results pinpoint the property required from an
oracle failure detector for supporting wait-free dynamic
storage. 
%(see conclusion in Section \ref{sec:discussion}).

\section{Model and Dynamic Storage Problem Definition}

\label{sec:models}

\begin{comment}
This section is organized as follow: 
In Section \ref{sub:arrival}, we define models with
infinitely many processes, and explain the difference among
them. In Section \ref{sub:notations}, we define the notation and
context we use, and in Section \ref{sub:dynamic}, we define
dynamic objects.
\end{comment}

\subsection{Preliminaries}
\label{sub:preliminaries}

We consider an asynchronous message passing system consisting of
an infinite set of processes $\Pi$. Each pair of processes is
connected by a communication link. Processes may fail by crashing
subject to restrictions on the number of failures given below.
A \emph{service} exposes a set of \emph{operations} to its
clients.

An \emph{algorithm} $A$ defines the behaviors of processes as
deterministic state machines, where state transitions are
associated with \emph{actions}, such as send/receive messages,
operation invoke/response, and process failures.
A \emph{global state} is mapping to states from system
components, i.e., processes and links. An \emph{initial
global state} is one where all processes are in initial
states, and all links are empty.
A \emph{run} of algorithm $A$ is a (finite or infinite) 
alternating sequence of global states and actions, beginning
with some initial global state, such that state
transitions occur according to $A$.
We use the notion of
\emph{time $t$ during a run $r$} to refer to the global state
incurred following the $t$\textsuperscript{th} action in $r$.
A \emph{run fragment} is a continuous subsequence of a
run.
An operation invoked before time $t$ in run $r$ is
\emph{complete} at time $t$ if its response event occurs before
time $t$ in $r$; otherwise it is \emph{pending} at time $t$.
We assume that runs are \emph{well-formed}
\cite{linearizability}, in that each process's first action is
an invocation of some operation, and a process does not invoke
an operation before receiving a response to its last invoked
one.
We say that operation $op_i$ \emph{precedes} operation $op_j$ in
a run $r$, if $op_i$'s response occurs before $op_j$'s
invocation  in $r$, and operations $op_i$ and $op_j$ are 
\emph{concurrent} in run $r$, if $op_i$ does not precede
$op_j$ and $op_j$ does not precede $op_i$ in $r$.
A \emph{sequential run} is one with no concurrent operations.
%An operation is pending in $r$ if it is invoked has no response
%in $r$.
Two runs are \emph{equivalent} if every process performs the
same sequence of invoke and response actions in both (with the
same return values).
%A \emph{sequential run} is one in which each
%invocation is immediately followed by its response.
%An \emph{execution of operation} $a$ that occurs in run $r$ is
%the run fragment beginning with its invocation and ending with
%its response.

\subsection{Dynamic register service}

We consider a dynamic MWMR register
service~\cite{dynastore}, which stores a value $v$ from a
domain $\mathbb{V}$, and offers an interface for invoking
\textit{read}, \textit{write}, and reconfiguration operations.
Initially, the register holds some initial value $v_0 \in
\mathbb{V}$.
We define \emph{Changes} to be the set $\{remove, add\} \times
\Pi$, and call any subset of Changes a \emph{set of changes}.
For example, $ \{ \langle add, p_3 \rangle, \langle remove, p_2
\rangle\}$ is a set of changes.
A \emph{reconfig} operation
takes as a parameter a set of changes and returns ``ok''.
We say that a change $w \in Changes$ completes before time $t$
in a run $r$ , if some \emph{reconfig(c)} completes in $r$
before time $t$ with $w \in c$.
We define $P_0 \subset \Pi$ to be the set of \emph{initial
processes} and say, by convention, that
\emph{reconfig}$(\{ \langle add,p \rangle |p \in P_0 \})$
completes at time 0.
We assume that $P_0$ is fixed and known to all.
%We say that a process $p$ was added in $r$ if $\langle add,p
%\rangle$ completed in $r$. 

\textbf{Notation}~~~
For every subset $w$ of $Changes$, 
the \emph{removal set} of $w$, denoted $w.remove$, is\\
$\{p_i| \langle remove,p_i \rangle \in w \}$; the \emph{join
set} of $w$, denoted $w.join$, is $\{ p_i|  \langle
add,p_i \rangle \in w\}$; and the \emph{membership} of $w$,
denoted $w.members$, is $w.join\setminus w.remove$.  
For a time $t$ in a run $r$, we 
denote by $V(t)$ the union of all sets $q$ s.t.\
\emph{reconfig(q)} completes before time $t$ in $r$.
A \emph{configuration} is a finite set of processes,
and the \emph{current configuration at
time $t$} is $V(t).membership$.
We define $P(t)$ to be the set of
\emph{pending changes} at time $t$ in run $r$, i.e., the set of
all changes included in pending reconfig operations, and
we denote by $F(t)$ the set of processes
that have failed before time $t$ in $r$, initially,
$F(0)=\{\}$.
For a series of sets $S(t)$, we define $S(*) \defeq \bigcup_{t
\in \mathbb{N}} S(t)$.

\textbf{Correct processes and fairness}~~~
A process $p$ is \emph{correct} if $p \in V(*).join \setminus
F(*)$.
A run $r$ is \emph{fair} if
every enabled action by a correct process eventually occurs,
and every message sent by a correct process $p_i$ to a correct
process $p_j$ is eventually received at $p_j$.
A process $p$ is \emph{active} if $p$ is correct and
$p \not\in P(*).remove$.
For simplicity, we assume that a process that has been removed
is not added again.

\textbf{Service specification}~~~
\emph{Atomicity}, also called
\emph{linearizabilty}~\cite{linearizability}, requires that for
every run, after adding some response actions and then removing
invocations that have no response, there exists an equivalent
sequential run that conforms with the operation precedence
relation, and satisfies the service's sequential specification.    
The sequential specification for the register service is as
follows: A read returns the
latest written value or $v_0$ if none was written.

\begin{comment}
An MWMR register is \emph{regular}~\cite{lamportRegular}, if
any \emph{read} returns either (1) one of the values written by a
concurrent \emph{write}, or (2) the latest value written before
the \emph{read} was invoked, or $\perp$ if no value was written.
\end{comment}

Lamport~\cite{lamportRegular} defines a \emph{safe} single-writer
register. Here, we generalize the definition to multi-writer
registers in a weak way in order to strengthen the
impossibility result:
An MWMR register is \emph{safe}, if in
every sequential run $r$ every \emph{read} $rd$ in $r$ 
returns the register's value when the read was invoked.

A \emph{wait-free} service 
guarantees that every active process's operation
completes, regardless of the actions of other processes.
A \emph{wait-free dynamic atomic storage} is a dynamic
storage service that satisfies atomicity and wait-freedom, and a
\emph{wait-free dynamic safe storage} is one that
satisfies safety and wait-freedom.

\textbf{Fault tolerance}~~~
We now specify conditions on when processes are allowed to fail.
First, we allow processes that are no longer part of the
current configuration's membership to be safely switched off.
To capture this property, we say that a
model is \emph{reconfigurable} if at any time $t$ in a run
$r$, any process in $V(t).remove$ can be in $F(t)$.
In other words, an adversary is allowed to crash any removed
process.
For our lower bound in Section \ref{sec:storageImpossibility},
we define in addition the \emph{minimal failure} condition:
whenever 
$(V(t).members\cup
P(t).join) \cap F(t) = \{\}$, at least one
process from $V(t).membership \cup P(t)$ can fail. In other
words, whenever no unremoved process is 
faulty, the adversary is allowed to fail at least one unremoved
process.

The above two conditions strengthen the adversary, allowing it
to fail processes in some scenarios. For our algorithm in
Section \ref{sec:diamondP}, we need to also restrict the
adversary, so as not to crash too many processes:
We say that a model allows \emph{minority failures} if at all
times $t$ in $r$, fewer than $|V(t).members \setminus
P(t).remove|/2$ processes out of $V(t).members\cup P(t).join$
are in $F(t)$. 

Notice that whenever $|V(t).members
\setminus P(t).remove| \geq 3$, the minority failure condition
allows minimal failure. 

\textbf{Suspicions}~~~
A fundamental property of an asynchronous system is that
failures cannot be accurately detected, and as long as processes
can fail, a correct process can be suspected, in the sense that
its failure is indistinguishable to other processes from slow
delivery of its messages.
To capture this property in other models (e.g., ones with oracles), we define
the following:

A process $p$ \emph{can be suspected at time $t$} if (a) $p$
has not failed before time $t$, (b) $p$ can fail at time $t$
according to the failure model; and (c) for any $t'>t$, every
run fragment lasting from time $t$ to time $t'$ where $p$ fails
at time $t$ is indistinguishable to all other processes from a
run fragment where $p$ is correct but all of its messages are
delayed from time $t$ to time $t'$.

Failure detectors like $\diamond S$ or $\Omega$ in a given
configuration guarantee that there is eventually a time $t$ s.t.\
there is \emph {one} process in the configuration that cannot be
suspected after time $t$. Nevertheless, other process may
continue to be suspected forever:

\begin{observation}

Consider an asynchronous model where processes are equipped with
$\Omega$ or $\diamond S$ in every configuration. Then in every
run, there is some process that can be suspected at any time
when its failure is allowed.

\end{observation}

Our impossibility result shows that wait-free
dynamic storage emulation is impossible in such models 
even if only reconfigurability and minimal failure are required.

\kern-1em
\section{Related Work}
\label{sec:relatedWork}

Previous works on asynchronous dynamic storage assume either
weak failure detectors like $\diamond S$ and
$\Omega$~\cite{gilbert2003rambo,gilbert2010rambo,RDS2009reconfigurable,baldoni2009implementing,SMRlamport2010reconfiguring,SMRvirtually},
or none at all~\cite{dynastore,ShraerDynadisc}.
Therefore, they are all subject to our impossibility result in
one way or another, as we now explain.

Our minority and reconfigurability failure conditions are based
on DynaStore's~\cite{dynastore,ShraerDynadisc} failure model,
with the difference that we distinguish between removed
processes and failed ones, and thus allow more failures.
In addition, as long as there 
are at least three members in each current or pending
configuration, minimal failures are allowed and so DynaStore is
subject to our impossibility, and indeed, 
guarantees liveness under the assumption that number of
reconfigurations is finite.

RAMBO~\cite{gilbert2010rambo,gilbert2003rambo} and
RDS~\cite{RDS2009reconfigurable} use
on consensus to agree on reconfigurations, while read and write
operations are asynchronous. They only discusses
liveness and fault tolerance in synchronous runs with bounded
churn and no guarantee on
reconfigurability~\cite{gilbert2010rambo}.
A similar liveness condition based on churn are used
in~\cite{gilbert2003rambo,baldoni2009implementing}. Therefore,
these algorithm do not contradict our impossibility result.

Reconfigurable
Paxos variants~\cite{SMRlamport2010reconfiguring,SMRvirtually} 
provide dynamic state machine replication, and
in turn implement dynamic atomic storage. These works are
subject to our impossibility result because they assume $\Omega$
(a leader) in every configuration. A configuration
may be changed, and accordingly a leader may be removed (and
then fail) before a process $p$ (with a pending operation) is
able to communicate with it.
Though a new leader is elected by $\Omega$ in the ensuing
configuration, this scenario may repeat itself indefinitely.
In Section \ref{sec:diamondP}, we augment state-machine
replication with helping based on a stronger failure detector in
order to avoid such scenarios.

A related impossibility proof~\cite{baldoni2009implementing}
shows that liveness is impossible to achieve with failures of
more than a minority in the current configuration, which
in some sense suggests that our minority failure condition is
tight.

\kern-1em
\section{Impossibility of Wait-Free Dynamic Safe Storage}
\label{sec:storageImpossibility}

In this section we prove that there is no implementation of 
wait-free dynamic safe storage in a reconfigurable model that
allows minimal failures and a correct process may be suspected
forever.

\begin{theorem}
\label{theorem:impossibility}

Consider an asynchronous model allowing reconfigurability and
minimal failures, where some correct process can be
suspected at any time when its failure is allowed by the
model.
Then there is no algorithm that emulates a wait-free dynamic
safe storage.
\end{theorem}

\begin{proof}[Proof (Theorem~\ref{theorem:impossibility})]
Assume by contradiction that such an algorithm $A$ exists. We
prove two lemmas about $A$.

\begin{myframe}
\begin{lemma}
\label{Lem:storageImpossibilityL1}

Consider a run $r$ of $A$ ending at time $t$, and
some process $p_i$ that can be suspected at time $t$.
Consider an active process $p_j$ in $r$ that invokes operation
$op$ at time $t$. 
Then there exists an extension of $r$ where (1) $op$ completes
at some time $t' > t$,(2) no process receives a message from
$p_i$ between $t$ and $t'$, and (3) $p_i$ does not fail. 

\end{lemma}

\begin{proof}[Proof (Lemma~\ref{Lem:storageImpossibilityL1})]

Consider a run fragment that begins at time $t$, in which
$p_i$ fails at time $t$ and all of its in transit messages
are lost. By wait-freedom, $op$ eventually
completes at some time $t'$. Since $p_i$ fails at time $t$
and all its outstanding messages are lost, in the run
fragment $\sigma_1$ starting from the global state at time $t$
and ending when $op$ is complete, no
process receives any message from $p_i$.
Now let $\sigma_2$ be
another run fragment lasting from time $t$ to time $t'$, in
which $p_i$ does not fail, but all of its messages are
delayed.
Recall that $p_i$ can be suspected at time $t$, therefore, 
$\sigma_1$ and $\sigma_2$ are indistinguishable to all
processes except $p_i$. Thus, $op$ returns also in $\sigma_2$.

\renewcommand{\qedsymbol}{\qedlemma{Lem:storageImpossibilityL1}}
\end{proof}

\begin{lemma}
\label{Lem:storageImpossibilityL2}

Consider a sequential run $r$ of $A$ ending at time $t$, where
some correct process $p_i$ can be suspected at any time in $r$
and there is some active process $p_j \neq p_i$ in $r$.
Assume that no process invokes $write(v_1)$ for some $v_1 \neq
v_0$ in $r$.
If we extend $r$ so that $p_i$ invokes $w=write(v_1)$ at
time $t$, and $w$ completes at some time $t' > t$, 
then in the run fragment between $t$ and $t'$, some process $p_k
\neq p_i$ receives a message sent by $p_i$. 

\end{lemma}

\begin{proof}[Proof (Lemma~\ref{Lem:storageImpossibilityL2})]

Assume by way of contradiction that \textit{w} completes at
some point $t'$, and in the run fragment between $t$
and $t'$ no process $p_k \neq p_i$ receives a message sent by
$p_i$.
Consider some other run $r'$ that is identical to $r$ until
time $t'$ except that $p_i$ does not invoke \textit{w} at time
$t$.
Now assume that process $p_j$ invokes a $read$
operation $rd$ at time $t'$ in $r'$. By the assumption,
$p_i$ can be suspected at $t'$.
Therefore, by Lemma \ref{Lem:storageImpossibilityL1}, there is
a run fragment $\sigma$ of $r'$ beginning at time $t'$, where
$rd$ completes at some time $t''$, and no process receives a
message from $p_i$ between $t'$ and $t''$.
Since no other process invokes $write(v_1)$ in $r'$, $rd$
returns some $v_2 \neq v_1$.
Now notice that all global states from time $t$ to time $t'$ in
$r$ and $r'$ are indistinguishable to all processes
except $p_i$.
Thus, we can continue the run $r$ with an invocation of read
operation $rd'$ by $p_j$ at time $t'$, and appending $\sigma$
to it.
Operation $rd'$ hence, completes and returns $v_2$. A
contradiction to safety.

\renewcommand{\qedsymbol}{\qedlemma{Lem:storageImpossibilityL2}}
\end{proof}
\end{myframe}

To prove the theorem, we construct an
infinite fair run $r$ in which a \textit{write}
operation of an active process never completes, in
contradiction to wait-freedom. An illustration of the run for
$n=4$ is presented in Figure \ref{fig:illustration1}.

%\begin{proof}{(Theorem \ref{theorem:impossibility})}
Consider some initial global state $c_0$, s.t.\
$P(0)=\{\}$ and $V(0).members=\{p_1\ldots p_n\}$.
By the assumption, there is some process $p$ that can be
suspected at any time $t$ when its failure is allowed.
Assume w.l.o.g. that this process is $p_1$,
and let it invoke write operation \textit{w} at
time 0. Let $t_1=0$. Now repeatedly do the following:

Let process $p_n$ invoke
\emph{reconfig(q)} where $q=\{\langle add,p_j \rangle | n+1 \leq
j \leq 2n-2\}$ at time $t_1$.
Since $(V(t_1).members\cup P(t_1).join) \cap
F(t_1) = \{\}$, $p_1$ can fail according to the minimal failures
condition, and therefore, by our assumption, can
be suspected at time $t_1$.
So by Lemma \ref{Lem:storageImpossibilityL1}, we can extend $r$
with a run fragment $\sigma_1$ ending at some time $t_2$ when
\textit{reconfig(q)} completes, no process $p_j
\neq p_1$ receives a message from $p_1$ in $\sigma_1$, and $p_1$
does not fail.

Then, at time $t_2$, $p_n$ invokes \emph{reconfig($q'$)}, where
$q'=\{\langle remove,p_j \rangle | 2\leq j \leq n-1\}$. 
Again, $(V(t_2).members\cup P(t_2).join) \cap
F(t_2) = \{\}$, and therefore, by our assumptions $p_1$
can be suspected at time $t_2$. And again, by Lemma
\ref{Lem:storageImpossibilityL1}, we can extend $r$ with a run
fragment $\sigma_2$ ending at some time $t_3$, when
\textit{reconfig(q)} completes, no process $p_j
\neq p_1$ receives a message from $p_1$ in $\sigma_2$, and
$p_1$ does not fail. 

Recall that we assume a reconfigurable model, so all the
processes in $V(t_3).remove$ can be now added to $F(t_3)$.
Therefore, let the process in $\{p_j \mid 2\leq j \leq n-1 \}$
fail at time $t_3$, and notice that the fairness condition does
not mandate that they receive messages from $p_1$. Next, allow
$p_1$ to perform all its enabled actions till some time $t_4$.
 
Now notice that at $t_4$, $|V(t_4).members|=n$, 
$P(t_4)=\{\}$, and $(V(t_4).members\cup
P(t_4).join) \cap F(t_4) = \{\}$. We can rename 
the processes in $V(t_4).members$ (except $p_1$) so that 
the process that performed the remove and add operations becomes
$p_2$, and all other get names in the range $p_3\ldots p_n$. 
We can then repeat the
construction above.
By doing so infinitely many
times, we get an infinite run $r$ in which $p_1$ is active
and no process ever receives a message from $p_1$. 
However, all of $p_1$'s enabled actions eventually occur.
Since no process except $p_1$ is correct in $r$, the run is
fair.
In addition, since $(V(t).members\cup
P(t).join) \cap F(t) = \{\}$ for all $t$ in $r$, by the minimal
failures condition, $p_1$ can fail at any time $t$ in $r$.
Hence, by the theorem's assumption, can be suspected at any
time $t$ in $r$.
Therefore, by Lemma \ref{Lem:storageImpossibilityL2}, $w$ does
not complete in $r$, and we get a violation of wait-freedom.

\renewcommand{\qedsymbol}{\qedthm{theorem:impossibility}}
\end{proof}

\kern-1em

\begin{figure}[ht]      
  \centering 
  \includegraphics[width=4in]{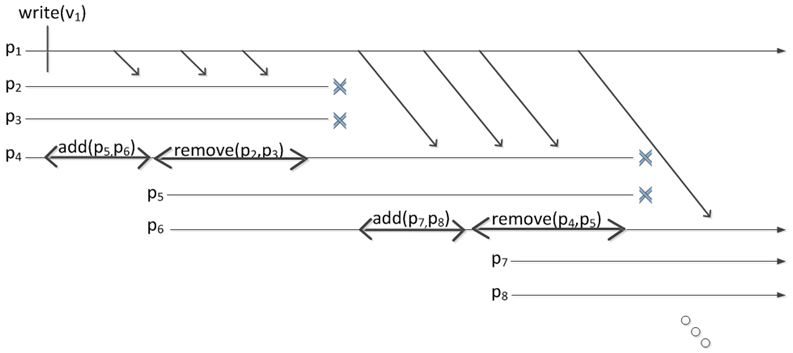}
  \caption[Illustration of the infinite run $r$ ]
   {Illustration of the infinite run for $n=4$.}
   \label{fig:illustration1}
\end{figure}

\kern-2em
\section{Oracle-Based Dynamic Atomic Storage}
\label{sec:diamondP}

We present an algorithm that circumvents
the impossibility result of Section
\ref{sec:storageImpossibility} using a failure detector.
In this section we a assume reconfigurable model with the
minority failure condition.
In Section \ref{sub:oracleDef}, we define a dynamic eventually
perfect failure detector. In Section
\ref{sub:oracleAlg}, we describe an algorithm, based on dynamic
state machine replication, that uses the failure detector to
implement a dynamic atomic MWMR register.
The algorithm's correctness is proven in Appendix
\ref{App:AppendixA}.

\subsection{Dynamic failure detector}
\label{sub:oracleDef} 

\begin{comment}
Dynamic failure detectors have been defined in previous
work~\cite{chockler2001failure,lin1999failure}. 
In our context, we define the following dynamic failure detector
$\Diamond P^D$.
Since the set of processes is potentially infinite we cannot
have the failure detector report the status of all processes as static
failure detectors typically do. Instead, it is queried
separately about each process. For each query, it answers either
\emph{fail} or \emph{ok}.
The dynamic eventually perfect failure detector $\Diamond P^D$
can be wrong for an unbounded period, but for each process, it
eventually returns a correct answer. 

Formally a $\Diamond p^D$ failure detector satisfies two
properties:

\begin{description}
  \item \textbf{Strong Completeness:} For each process $p_i$
  that fails at time $t_i,$ there is a time $t>t_i$ s.t. the
  failure detector answers \emph{fail} to every
  query about $p_i$ after time $t$.

  \item \textbf{Eventual Strong Accuracy:} There exists a time
  $t$ s.t.\ the failure detector answers
  \emph{ok} to every query at time $t'>t$ about a process that
  is active at time $t'$.
\end{description}

Notice that this failure detector gives less information than 
one that returns a full list of correct or faulty processes,
because, in dynamic systems, the latter gives additional
information about processes that could have been unknown
to the inquiring process.
\end{comment}

Since the set of processes is potentially infinite, we cannot
have the failure detector report the status of all processes as
static failure detectors typically do. Dynamic failure detectors
addressing this issue have been defined in previous works, either
providing a set of
processes that have been excluded from or included into the
group~\cite{lin1999failure}, or assuming that there is
eventually a fixed set of participating
processes~\cite{chockler2001failure}.

In our model, we do not assume that there is eventually a 
fixed set of participating processes, as the number of
reconfigurations can be infinite. And we do not want the failure
detector to answer with a list of processes, because in
dynamic systems, this gives additional information about 
participating processes that could have been unknown to the
inquiring process. Instead, our dynamic failure detector
is queried separately about each process.
For each query, it answers either \emph{fail} or \emph{ok}.
It can be wrong for an unbounded period, but for each process,
it eventually returns a correct answer. 

Formally, a \emph{dynamic eventually perfect} failure detector,
$\Diamond P^D$, satisfies two properties:

\begin{description}
  \item \textbf{Strong Completeness:} For each process $p_i$
  that fails at time $t_i$, there is a time $t>t_i$ s.t. the
  failure detector answers \emph{fail} to every
  query about $p_i$ after time $t$.

  \item \textbf{Eventual Strong Accuracy:} There exists a time
  $t$, called the \emph{stabilization time}, s.t.\ the failure
  detector answers \emph{ok} to every query at time $t'>t$ about
  a correct process that was added before time $t'$.
  
  \begin{comment}
  \item \textbf{Eventual Strong Accuracy:} For each active
  forever process $p$, there exists a time $t$ s.t.\ the failure
  detector answers \emph{ok} to every query at time $t'>t$ about
  $p$.
  \end{comment}
  
\end{description}
 
\subsection{Dynamic storage algorithm}
\label{sub:oracleAlg}

\subsubsection{Algorithm overview}

\textbf{State machine emulation of register}
We use a state machine \emph{sm} to emulate
a wait-free atomic dynamic register, \emph{DynaReg}. Every
process has a local replica of \emph{sm}, and we use
consensus~\cite{consensus1980reaching} to agree on \emph{sm}'s
state transitions. Notice that each process is equipped with a
failure detector FD of class $\Diamond P^D$, so consensus is
solvable under the assumption of a correct majority in a given
configuration.

Each consensus runs in a given configuration $c$, exposes a
\emph{propose} operation, and responds with \emph{decide},
satisfies the following properties: By \emph{Uniform
Agreement}, every two decisions are the same. 
By \emph{Validity}, every decision was previously proposed by
one of the processes in $c$. 
By \emph{Termination}, if a majority of $c$ is correct, then
eventually every correct processes in $c$ decides.
We further assume that a consensus instance does not decide
until a majority of the members of the configuration propose in
it.

The $sm$ (presented in lines
\ref{line:smBegin}-\ref{line:smEnd} in Algorithm
\ref{algOracle:definitions}) keeps track of dynaReg's value in a
variable $val$, and the configuration in a variable $cng$,
containing both a list of processes, $cng.mem$, and a set of
removed processes, $cng.rem$.
Write operations change $val$, and reconfig operations change
$cng$.
A consensus decision may bundle a number of operations to
execute as a single state transition of $sm$. The number of
state transitions executed by $sm$ is stored in the variable
$ts$.
Finally, the array $lastOps$ maps every process $p$ in $cng.mem$
to the sequence number (based on $p$'s local count) of $p$'s last
operation that was performed on the emulated DynaReg together
with its result.

Each process partakes in at most one consensus at a time; this
consensus is associated with timestamp $sm.ts$ and runs in
$sm.cng.mem$.
In every consensus, up to \emph{$|sm.cng.mem|$}
ordered operations on the emulated DynaReg are agreed upon, and 
\emph{sm}'s state changes according to the agreed operations. 
A process's $sm$ may change either when consensus decides or
when the process receives a newer $sm$ from another process, in
which case it skips forward. So \emph{sm} goes through the same
states in all the processes, except when skipping forward.
Thus, for every two processes $p_k,p_l$, if
$sm_k.ts=sm_l.ts$, then $sm_k=sm_l$. (A subscript $i$ indicates
the variable is of process $p_i$.)

\textbf{Helping}
The problematic scenario in the impossibility proof of Section
\ref{sec:storageImpossibility} occurs because of endless
reconfigurations, where a slow process is never able to
communicate with members of its configuration before they are
removed.
In order to circumvent this problem, we use the FD to implement
a helping mechanism.
When proposing an operation, process $p_i$ tries to help other
processes in two ways: first, it helps them
complete operations they may have successfully proposed in
previous rounds but have not learned about their outcome; and
second, it proposes their new operations.
First, it sends its \emph{sm} to all other processes in
\emph{$sm_i.cng.mem$}, and waits for each to reply with its
latest invoked operation. Then $p_i$
proposes all the operations together.
Processes may fail or be removed, so $p_i$ cannot wait for
answers forever.
To this end, we use the FD.
For every process in \emph{$sm_i.cng.mem$} that has not been
removed, $p_i$ repeatedly inquires FD and waits either for a
reply from the process or for an answer from the FD that the
process has failed.
Notice that the strong completeness property guarantees that
$p_i$ will eventually continue, and strong accuracy guarantees
that every slow active process will eventually receive
help in case of endless reconfigurations.

Nevertheless, if the number of reconfigurations is finite, it
may be the case that some slow process is not familiar with any
of the correct members in the current configuration, and no
other process performs an operation (hence, no process is
helping).
To ensure progress in such cases, every correct process
periodically sends its $sm$ to all processes in its
$sm.cng.mem$

\textbf{State survival}
Before the reconfig operation can complete, the new $sm$
needs to propagate to a majority of the new configuration, in
order to ensure its survival.
Therefore, after executing the state transition, $p_i$ sends
$sm_i$ to $sm_i.cng$ members and waits until it either receives
acknowledgements from a majority or learns of a newer $sm$.
Notice that in latter the case, consensus in $sm_i.cng.mem$ has
decided, meaning that at least a majority of $sm_i.cng.mem$ have
participated in it, and so have learned of it.

\textbf{Flow example}
The algorithm flow is illustrated in Figure
\ref{fig:FlowAlgOracle}. In this example, a slow process $p_2$
invokes operation $op_{21}$ before
the FD's stabilization time, ST.
Process $p_1$ invokes operation $op_{11}=\langle add, p_3 
\rangle$ after ST. It first sends \emph{helpRequest} to $p_2$
and waits for it to reply with \emph{helpReply}. Then it
proposes $op_{21}$ and $op_{11}$ in a consensus. When
\emph{decide} occurs, $p_1$ updates its $sm$, sends it to all
processes, and waits for majority. 
%This is done in order to
%ensure that $sm$'s state has propagated to the new
%configuration.
Then $op_{11}$ returns and $p_1$ fails before $p_2$ receives
its update message. Next, $p_3$ invokes a \emph{reconfig}
operation, but this time when $p_2$ receives \emph{helpRequest}
with the up-to-date $sm$ from $p_3$, it notices that its
operation has been performed, and $op_{21}$ returns.

\begin{figure}[H]         
  \centering  
  \includegraphics[width=6.2in]{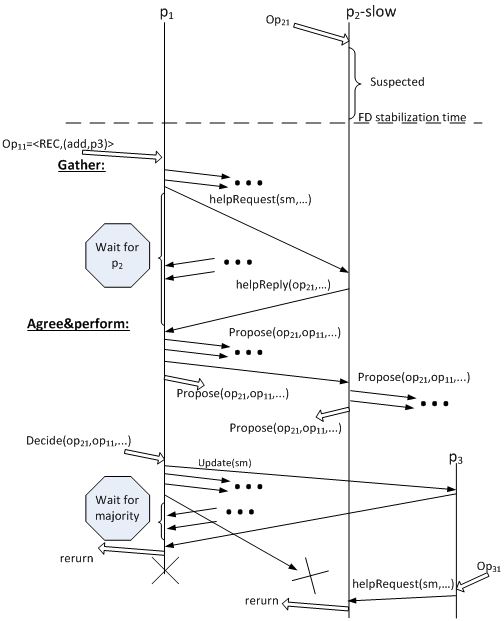} 
  \caption[Algorithm flow]
   {Flow illustration: process $p_2$ is slow.
   After stabilization time, process $p_1$ helps it by
   proposing its operation. Once $p_2$'s operation is decided,
   it is reflected in every up-to-date $sm$.
   Therefore, even if $p_1$ fails before informing $p_2$, $p_2$
   receives from the next process that performs an operation,
   namely, $p_3$, an $sm$ that reflects its operation, and thus
   returns.
   Line arrows represent messages, and block arrows represent
   operation or consensus invocations and responses.}
   \label{fig:FlowAlgOracle}
\end{figure}

\subsubsection{Detailed description}

The data structure of process $p_i$ is given in Algorithm
\ref{algOracle:definitions}. First, $sm_i$, is described above.
Integer $opNum_i$ holds the sequence number of $p_i$'s current operation; $ops_i$ is a
set that contains operations that 
need to be completed for helping; the flag $pend_i$ is a
boolean that indicates whether or not $p_i$ is participating in
an ongoing consensus; and $myOp_i$ is the latest operation
invoked at $p_i$.

The algorithm of process $p_i$ is presented in Algorithms
\ref{algOracle:operation} and \ref{algOracle:messages}.
We execute every event handler, (operation invocation, message
receiving, and consensus decision), atomically excluding wait
instructions; that is,
other event handlers may run after the handler completes or
during a wait (lines
\ref{lineOracle:gatheE},\ref{line:wait2},\ref{lineOracle:op7} in
Algorithm \ref{algOracle:operation}).
The algorithm runs in two phases. The first, \emph{gather}, is
described in Algorithm \ref{algOracle:operation} lines
\ref{lineOracle:gatheB}-\ref{lineOracle:gatheE} and in Algorithm
\ref{algOracle:messages} lines
\ref{lineOracle:gatherHReqB}-\ref{lineOracle:gatherHreqRE},
\ref{lineOracle:gatherHrepRB}-\ref{lineOracle:gatherHrepRE}.
Process $p_i$ first increases its operation number $opNum_i$,
writes $op$ together with $opNum_i$ to the set of operations
$ops_i$, and sets $myOp_i$ to be $op$. Then it sends
$\langle$``helpRequest''$,\ldots \rangle$ to every member of
$A=sm_i.cng.mem$ (line \ref{lineOracle:op1}), and waits for each
process in $A$ that is not suspected
by the FD or removed to reply with
$\langle$``helpReply''$,\ldots \rangle$. Notice that $sm_i$ may
change during the wait because messages are handled, and $p_i$
may learn of processes that have been removed.

When $\langle$``helpRequest''$,num,sm \rangle$ is received by
process $p_j \neq p_i$, if the received \emph{sm} is 
newer than $sm_j$, then process $p_j$ adopts \emph{sm} and
abandons any previous consensus. Either way, $p_j$ sends
$\langle$``helpReply''$,\ldots\rangle$ with its current
operation $myOp_j$ in return.

Upon receiving $\langle$``helpReply''$,opNum_i,op,num \rangle$
that corresponds to the current operation number $opNum_i$,
process $p_i$ adds the received operation \emph{op}, its
number \emph{num}, and the identity of the sender to the set
$ops_i$.

At the end of this phase, process $p_i$ holds a set of
operations, including its own, that it tries to agree on in the
second phase.
Note that $p_i$ can participate in at most one consensus per
timestamp, and its propose might end up not being the decided
one, in which case it may need to propose the same operations
again.
Process $p_i$ completes $op$ when
it discovers that $op$ has been performed in $sm_i$, whether by
itself or by another process.

The second phase appears in Algorithm
\ref{algOracle:operation} lines
\ref{lineOracle:agreeB}-\ref{lineOracle:agreeE}, and in
Algorithm \ref{algOracle:messages} lines
\ref{lineOracle:decideB}-\ref{lineOracle:decideE},
\ref{lineOracle:proposeB}-\ref{lineOracle:proposeE}.
In line \ref{lineOracle:agreeB}, $p_i$ checks if its operation
has not been completed yet. In lines
\ref{lineOracle:op4} to \ref{lineOracle:op5}, it waits until it
does not participate in any ongoing consensus
($pend_i$=\emph{false}) or some other process helps it
complete $op$.
Recall that during a wait, other events can be handled. 
So if a message with an up-to-date $sm$ is received during the
wait, $p_i$ adopts the $sm$.
In case $op$ has been completed in $sm$, $p_i$ exits the main
while (line \ref{lineOracle:op5}).
Otherwise, $p_i$ waits until either it does not participate in
any ongoing consensus. This can be the
case if (1) $p_i$ has not proposed yet, (2) a message with
a newer $sm$ was received and a previous consensus was
subsequently abandoned, or (3) a $decide$ event has been
handled.
In all cases, $p_i$ marks that it now
participates in consensus in line \ref{lineOracle:op6}, prepares
a new request $Req$ with the operations in $ops_i$ that
have not been performed yet in $sm_i$ in
line \ref{lineOracle:op7}, proposes $Req$ in the consensus
associated with $sm_i.ts$, and sends
$\langle$``propose''$,\ldots\rangle$ to all the members of
$sm_i.cng.mem$.

When $\langle$``propose''$,sm,Req\ldots \rangle$ is received by
process $p_j \neq p_i$,
if the received $sm$ is more updated
than $sm_j$, then process $p_j$ adopts $sm$, abandons any
previous consensus, proposes $Req$ in the consensus associated
with $sm.ts$, and forwards the message to all other members of
$sm_j.cng.mem$.
The same is done if $sm$ is identical to $sm_j$ and $p_j$ has
not proposed yet in the consensus associated with
$sm_j.ts$.
Otherwise, $p_j$ ignores the message.

The event $decide_i(sm.cng,sm_i.ts,Req)$ indicates a
decision in the consensus associated with
$sm_i.ts$.
When this occurs, $p_i$ performs all the operations in
$Req$ and changes $sm_i$'s state.
It sets the value of the emulated DynaReg, $sm_i.value$, to be
the value of the \emph{write} operation of the process with the
lowest id, and updates $sm_i.cng$ according to
the \emph{reconfig} operations.
In addition, for every 
$\langle p_j,op,num \rangle \in Req$, $p_i$ writes to
$sm_i.lastOps[j]$, $num$ and $op$'s response, which is ``ok'' in
case of a \emph{write} or a \emph{reconfig}, and $sm_i.value$ in
case of a \emph{read}.
Next, $p_i$ increases $sm_i.ts$ and sets $pend_i$
to false, indicating that it no longer participates in any
ongoing consensus.

Finally, after $op$ is performed, $p_i$ exits the main while. If
$op$ is not a \emph{reconfig} operation, then $p_i$ returns the
result, which is stored in $sm_i.lastOps[i].res$. Otherwise,
before returning, $p_i$ has to be sure that a majority of
$sm_i.cng.mem$ receives $sm_i$. It sends
$\langle$``update''$,sm,\ldots\rangle$ to all the processes in
$sm_i.cng.mem$ and waits for $\langle$``ACK''$,\ldots\rangle$
from a majority of them.
Notice that it may be the case that there is no such correct
majority due to later reconfigurations and failures, so,
$p_i$ stops waiting when a more updated $sm$ is received, which
implies that a majority of $sm_i.cng.mem$ has already received
$sm_i$ (since a majority is needed in order to solve consensus).

Upon receiving $\langle$``update''$,sm,num\rangle$ with a new
$sm$ from process $p_i$, process $p_j$ adopts $sm$ and abandons
any previous consensus. In addition, if $num \neq \perp$, $p_j$
sends $\langle$``ACK''$,num\rangle$ to $p_i$ (Algorithm
\ref{algOracle:messages} lines
\ref{lineOracle:updateB}-\ref{lineOracle:updateE}).

Beyond handling operations, in
order to ensure progress in case no operations are invoked from
some point on, every correct process periodically sends
$\langle$``update''$,sm,\perp\rangle$ to all processes in its
$sm.cng.mem$ (Algorithm \ref{algOracle:operation} line
\ref{lineOracle:periodically}).

In Appendix \ref{App:AppendixA}, we prove that the algorithm
satisfies atomicity and wait-freedom.

\begin{algorithm}[H]
 \caption{Data structure of process $p_i$}  
 \label{algOracle:definitions}
\begin{algorithmic}[1]
\small

\State $sm_i.ts \in \mathbb{N}$, initially 0
\label{line:smBegin}
\State $sm_i.value \in\mathcal{V}\cup \{\perp\}$, initially
$\perp$
\State $sm_i.cng=\langle mem,rem \rangle$,
\Statex where $mem,rem \subset \Pi$, initially $\langle P_0, \{\}
\rangle$, where $P_0 \subset \Pi$

\State $sm_i.lastOps$ is a vector of size
$|sm_i.cng.mem|$, 
\Statex where $\forall p_j \in sm_i.cng.mem$,
$sm_i.lastOps[j] = \langle num,res \rangle$,
\Statex where $num \in
\mathbb{N}$, $res \in \mathcal{V}\cup \{\perp,$``ok''$\})$,
initially $\langle 0, \perp \rangle$
\label{line:smEnd}

\State $pend_i \in$ \emph{\{true,false\}}, initially \emph{false}
\State $opNum_i \in \mathbb{N}$, initially 0
%\Comment{the selectors are \emph{type,value,changes}}
\State $ops_i \subset \Pi \times operation \times
\mathbb{N}$ , initially $\{ \}$
\Statex where $operation = \{\langle RD\rangle,\langle
WR, value \rangle,\langle REC, changes \rangle, \perp
\}$, 
\Statex where $value \in \mathcal{V}$, $changes \subset
\{add,remove \} \times \Pi$, and \emph{type} can be
$RD,WR,REC$
\State $myOp_i \in operation$, initially $\perp$

\end{algorithmic}
\end{algorithm}

\begin{algorithm}[H]
% \floatname{algorithm}{ALGORITHM:}  
% \renewcommand{\thealgorithm}{}  
 \caption{Process $p_i$'s algorithm: performing operations}  
\label{algOracle:operation}
\begin{algorithmic}[1]
\small

\Upon{invoke operation($op$)}
	\State $opNum_i \leftarrow opNum_i +1$
		\Comment{ \textbf{phase 1:} \textit{gather}}
		\label{lineOracle:gatheB}
	\State $ops_i \leftarrow \{ \langle p_i, op, opNum_i \rangle \}$
	\State $myOp_i \leftarrow  op$
	\State $A \leftarrow sm_i.cng.mem$
	\State \textbf{for all} $p \in A$ send $\langle
	$``helpRequest''$,opNum_i,sm_i\rangle$ to $p$
		\label{lineOracle:op1} 
	\State  \textbf{for all} $p \in A$ \textbf{wait}
	for $\langle$``helpReply''$,opNum_i,\ldots\rangle$
	from $p$ or $p$ is suspected or $p \in sm_i.cng.rem$
	\label{lineOracle:gatheE}
	%\label{line:wait1}
	\State \textbf{while} $sm_i.lastOps[i].num \neq opNum_i$
		\label{lineOracle:agreeB}
		\Comment{\textbf{phase 2:} \textit{agree\&perform}} 
	\State \hspace*{0.5cm} \textbf{while} (\emph{$pend_i$})
		\label{lineOracle:op4}  
	\State \hspace*{1cm} \textbf{wait} until $\neg pend_i$ or 
	$sm_i.lastOps[i].num=opNum_i$ 
		\label{line:wait2}
	\State \hspace*{1cm} \textbf{if} $sm_i.lastOps[i].num=opNum_i$ 
%	\textbf{then} goto line \ref{line:return}
	\textbf{then} \textbf{goto} line \ref{lineOracle:goto}
		\label{lineOracle:op5} 
	\State \hspace*{0.5cm} $pend_i \leftarrow true$ 
		\label{lineOracle:op6}
%	\State \hspace*{0.5cm} $Req \leftarrow prepareRequest()$
	\State \hspace*{0.5cm} $Req \leftarrow \{ \langle p_j,op,num
	\rangle \in ops_i$ $|$ $num > sm_i.lastOps[j].num \}$
		\label{lineOracle:op7}
	\State \hspace*{0.5cm} $propose(sm_i.cng,sm_i.ts, Req)$
		\label{lineOracle:propose}
	\State \hspace*{0.5cm} \textbf{for all} $p \in sm_i.cng.mem$
	send $\langle$``propose''$,sm_i,Req \rangle$ to $p$
	
	\State \textbf{if} $op.type =  REC $
	\label{lineOracle:goto}
	\State \hspace*{0.5cm} $ts \leftarrow sm_i.ts$
	\State \hspace*{0.5cm} \textbf{for all} $p \in sm_i.cng.mem$
	send $\langle$``update''$,sm_i,opNum_i \rangle$ to $p$\
	\label{line:oracleMajRec}
	\State \hspace*{0.5cm} \textbf{wait} for
	$\langle$``ACK''$,opNum_i\rangle$  from  majority
	of $sm_i.cng.mem$ or $sm_i.ts > ts$
	\label{lineOracle:op7}
	
	\State \textbf{return} $sm_i.lastOps[i].res$
	\label{lineOracle:op8}

\EndUpon
\label{lineOracle:agreeE} 

\Statex

\State \textbf{periodically:}
\State \hspace*{0.5cm} \textbf{for all} $p \in sm_i.cng.mem$
send $\langle $``update''$,sm_i,\perp\rangle$ to $p$
	\label{lineOracle:periodically}

\end{algorithmic}
\end{algorithm}

\begin{algorithm}[H]
% \floatname{algorithm}{ALGORITHM:}  
% \renewcommand{\thealgorithm}{}  
 \caption{ Process $p_i$'s algorithm: event handlers}  
 \label{algOracle:messages}
\begin{algorithmic}[1]
\small

\Statex

  \Upon{$decide_i(sm_i.cng,sm_i.ts, Req)$}
 	\label{lineOracle:decideB}
 	
 	\State $W \leftarrow \{\langle p,value \rangle |
 	\langle p,\langle WR,value \rangle, num \rangle \in Req \}$
 	
 	\State \textbf{if} $W \neq \{\}$
 	
 	\State \hspace*{0.5cm} $sm_i.value \leftarrow$ value associated
 	with smallest $p$ in $W$
  	
	\State \textbf{for all} $\langle p,op,num \rangle \in Req$
		\State \hspace*{0.5cm} \textbf{if} \emph{op.type = WR}
			\State \hspace*{1cm} $sm_i.lastOps[j] \leftarrow
			\langle num,$``ok''$\rangle$ 
		\State \hspace*{0.5cm} \textbf{else if}
		\emph{op.type = RD} 
			\State \hspace*{1cm} $sm_i.lastOps[j] \leftarrow
			\langle num, sm_i.value \rangle$
		\State \hspace*{0.5cm} \textbf{else}
			\State \hspace*{1cm} $sm_i.cng.mem \leftarrow \{p|
			(p \in sm_i.cng.mem$ \& $ \langle remove,p \rangle
			\notin op.changes)$ $||$
			\Statex \hspace*{1.6cm} $(\langle add,p
			\rangle \in op.changes$ \& $p \notin sm_i.cng.rem )\}$
			\label{lineOracle:dec1}
			\State \hspace*{1cm} $sm_i.cng.rem \leftarrow \{p| (p
			\in sm.cng.rem$ $||$ $\langle remove,p \rangle
			\in Changes \}$
			\label{lineOracle:dec2}
			\State \hspace*{1cm} $sm_i.lastOps[j] \leftarrow
			\langle num,$``ok''$\rangle$
			
			\State $sm_i.ts \leftarrow sm_i.ts +1$
			\State $pend_i \leftarrow$ \emph{false} 
 
  \EndUpon
  \label{lineOracle:decideE}

\Receiving{$\langle$``propose''$,sm,Req \rangle$ from $p_j$}
\label{lineOracle:proposeB}
	\State \textbf{if} ($sm_i.ts > sm.ts$) or ($sm_i.ts = sm.ts$
	\& $pend_i = true$) \textbf{then} return
	
	\State  $sm_i \leftarrow sm$ 
	\State  $pend_i \leftarrow true$ 
	\State  $propose(sm_i.cng,sm_i.ts, Req)$
	\State  \textbf{for all} $p \in sm_i.cng.mem$
	send $\langle$``propose''$,sm_i,Req \rangle$ to $p$
\EndReceiving
\label{lineOracle:proposeE}
  
%\Statex
  
\Receiving{$\langle$``helpRequest''$,num,sm\rangle$ from $p_j$}
	\label{lineOracle:gatherHReqB}
	\State \textbf{if} $sm_i.ts < sm.ts$ \textbf{then} 
	\State \hspace*{0.5cm} $sm_i \leftarrow sm$
	\State \hspace*{0.5cm} $pend_i \leftarrow$ \emph{false}
	\State send  $\langle$``helpReply''$,num,myOp_i,opNum_i \rangle$
  
\EndReceiving
\label{lineOracle:gatherHreqRE}

%\Statex

\Receiving{$\langle$``helpReply''$,opNum_i,op,num \rangle$
from $p_j$} 
\label{lineOracle:gatherHrepRB}
	\State $ops_i \leftarrow ops_i \cup \langle p_j,op,num \rangle $
\EndReceiving
\label{lineOracle:gatherHrepRE} 

%\Statex

\Receiving{$\langle$``update''$,sm,num \rangle$ from $p_j$} 
\label{lineOracle:updateB}
	\State \textbf{if} $sm_i.ts < sm.ts$ \textbf{then} 
	\State \hspace*{0.5cm} $sm_i \leftarrow sm$
	\State \hspace*{0.5cm} $pend_i \leftarrow$ \emph{false}
	\State \textbf{if} $num \neq \perp$ \textbf{then}
	send $\langle$``ACK''$,num\rangle$ to $p_j$ 
\EndReceiving 
\label{lineOracle:updateE}

\end{algorithmic}
\end{algorithm}

\section{Conclusion}
\label{sec:discussion}

We proved that in an asynchronous reconfigurable
model allowing at least one failure,
and no restriction on the number of reconfigurations, there is
no emulation of dynamic wait-free storage. This is true even for
safe storage, and even if processes are equipped with $\Omega$
or $\diamond S$ failure detectors, which allow them to solve
consensus in every configuration.
We further showed how to circumvent this result using a
dynamic eventually perfect failure detector: we presented an
algorithm that uses such a failure detector in order to emulate a wait-free
dynamic atomic MWMR register. 

Our results thus draw a distinction between models where correct
processes can be suspected at any time (as long as they may
fail), and ones where false suspicions eventually cease.

\newpage
\appendix
\section{Correctness Proof}
% the \\ insures the section title is centered below the phrase: AppendixA
\label{App:AppendixA}

Here we prove the correctness of our algorithm (Section
\ref{sec:diamondP}).

\subsection{Atomicity}

Every operation is uniquely defined by the process that
invoked it and its local number. During the proof we refer to
operation $op$ invoked by process $p_i$ with local number
$opNum_i$ as the tuple $\langle p_i,op,opNum_i\rangle$, or simply
as $opNum_i$. We begin the proof with three lemmas that link
completed operation to $sm$ states.

\begin{lemma}
\label{lem:oracle1}

Consider  operation $op$ is invoked by some process $p_i$ in $r$
with local number $opNum_i$. If $op$ returns in $r$ at
time $t$, then there is at least one request \emph{Req} that
contains $\langle p_i,op,opNum_i\rangle$ and has been chosen in 
a consensus in $r$ before time $t$.   

\end{lemma}

\begin{proof}

Operation $op$ cannot return until $sm_i.lastOps[i].num=opNum_i$
(line \ref{lineOracle:agreeB} or \ref{line:wait2} in
Algorithm \ref{algOracle:operation}). 
Processes update $sm$ during a decide handler, or when newer
$sm$ is received. Easy to show by induction that some process
$p_j$ writes $opNum_i$ to $sm_i.lastOps[i].num$ during a
decide handler. According to the run of the decide
handler, $opNum_i$ is written to $sm.lastOps[i].num$ only if the
chosen request in the
corresponding consensus contains $\langle
p_i,op,opNum_i\rangle$.

\end{proof}

\begin{lemma}
\label{lem:oracle3}

For every two processes $p_i,p_j$. Let $t$ be a time in $r$ in
which neither $p_i$ or $p_j$ executing decide handler. Then at
time $t$, if $s.m_i.ts=sm_j.ts$, then $sm_i=sm_j$.

\end{lemma}

\begin{proof}

We prove by induction on timestamps. Initially, all 
correct processes have the same $sm$ with timestamp 0.
Now consider timestamp $TS$, and assume that for every two
processes $p_i,p_j$ at any time not during the
execution of decide handlers, if $sm_i.ts=sm_j.ts=TS$,
then $sm_i=sm_j$. Processes increase their $sm.ts$ to $TS+1$
either at the end of a \emph{decide} handler associated with
$TS$ or when they receive a message with $sm$ s.t. $sm.ts=TS+1$.
By the agreement property of consensus and by the
determinism of the algorithm, all the processes that perform the
\emph{decide} handler associated with $TS$, perform the same
operations, and therefore move $sm$ (at the end of the handler)
to the same state.
It is easy to show by induction that all the processes
that receive a message with $sm$ s.t. $sm.ts=TS+1$, receive the
same $sm$. The lemma follows.

\end{proof}

\begin{observation}
\label{obs:mon}

For any two states $sm_1$, $sm_2$, and for any process $p_i$ in
a run $r$,
if $sm_1.ts \geq sm_2.ts$, then $sm_1.lastOps[i].num \geq
sm_2.lastOps[i].num$.

\end{observation}

\begin{proof}

Easy to show by induction.

\end{proof}

\begin{lemma}
\label{lem:oracle4}

Consider operation $op$ invoked in $r$ by some process $p_i$
with local number $opNum_i$ in $r$. Then $op$ is part of at
most one request that is chosen in a consensus in $r$. 

\end{lemma}

\begin{proof}

Assume by way of contradiction that $op$ is part of more than
one request that is chosen in a consensus in $r$. Now consider
the earliest one, $Req$, and assume that it is chosen in a
consensus associated with timestamp $TS$. At the end of the
\emph{decide} handler associated with timestamp $TS$,
$sm.lastOps[i].num=opNum_i$ and  it is increased to $TS+1$.
Thus, by Lemma \ref{lem:oracle3} $sm.lastOps[i].num=opNum_i$
holds for every $sm$ s.t.\ $sm.ts=TS+1$. 
Consider now the
next request, $Req_1$, that contains $op$, and is chosen
in a consensus. Assume that this consensus associated with ts
$TS'$, and notice that $TS' > TS$.
By the validity of consensus, this request is
proposed by some process $p_j$, when $sm_j.ts$ is equal to
$TS'$. 
By Observation \ref{obs:mon}, $sm.lastOps[i].num=opNum_i$ holds
for all $sm$ s.t.\ $sm.ts=TS'$, and therefore $p_j$ does not
enter $op$ to $Req_1$ (line \ref{lineOracle:op7} in Algorithm
\ref{algOracle:operation}). A contradiction.

\end{proof}

Based on the above lemmas, we can define, for each run $r$,
a linearization $\sigma_r$, where operations are ordered as they
are chosen for execution on $sm$'s in $r$.

\begin{definition}

For a run $r$, we define the sequential
run $\sigma_r$ to be the sequence of operations decided in
consensus instances in $r$,
ordered by the order of the chosen requests they are part of in
$r$.
The order among operations that are part of the same chosen
request is the following: first all  \emph{writes}, then all
\emph{reads}, and finally, all \emph{reconfig}
operations.
Among each type, operations are ordered by the process ids of the
processes that invoked them, from the highest to the lowest.

\end{definition}

\begin{corollary}
\label{cor:wellDef}

For every run $r$, the sequential execution $\sigma_r$ is well
defined.
That is, $\sigma_r$ contains every completed operation in
$r$ exactly once, and every invoked operation at most once.

\end{corollary}

\begin{comment}
\begin{proof} 

Notice that an operation cannot appear more the once in a
request,so together with Lemmas \ref{lem:oracle1} and
\ref{lem:oracle4}, every completed operation in $r$ is part of
exactly one chosen request.
Therefore every completed operation in $r$ appears exactly one
time in $\sigma_r$.

\end{proof}
\end{comment}

In order to prove atomicity it remains to
show that (1) $\sigma_r$ preserves $r$'s real time order; and (2)
every \emph{read} operation $rd$ in $r$ returns the value that
was written by the last \emph{write} operation that precedes
$rd$ in $\sigma_r$, or $\perp$ if there no such operation.

\begin{lemma}
\label{lem:oracle2}

If operation $op_1$ returns before operation $op_2$ is invoked in
$r$, then $op_1$ appears before $op_2$ in $\sigma_r$.

\end{lemma}

\begin{proof}

Operation $op_1$ returns before operation $op_2$ is invoked in
$r$. Therefore By Lemma \ref{lem:oracle1}, $op_1$ is part of a
request $Req_1$ that is chosen in a consensus before $op_2$ is
invoked, and thus $op_2$ cannot be part of $Req_1$ or any other
request that is chosen before $Req_1$. Hence $op_1$ appears
before $op_2$ in $\sigma_r$.

\end{proof}

\begin{lemma}
\label{lem:oracle5}

Consider read operation $rd$ invoked by some process
$p_i$ with local number $opNum_i$ in $r$,
which returns a value $v$. Then $v$ is written by the last
write operation that precedes $rd$ in $\sigma_r$, or
$v=\perp$ if there is no such operation.

\end{lemma}

\begin{proof}

By Lemmas \ref{lem:oracle1} and \ref{lem:oracle4}, $rd$ is part
of exactly one request $Req_1$ that is chosen in a consensus,
associated with some timestamp $TS$. Thus
$sm.lastOps[i]$ is set to $\langle opNum_i, sm.value \rangle$ in the
\emph{decide} handler associated with $TS$, denote the value of
$sm.value$ at this point to be $val$.
By Lemma \ref{lem:oracle3}, $sm.lastOps[i]=\langle opNum_i, val
\rangle$ for all $sm$ s.t.\ $sm.ts=TS+1$.
By Lemma \ref{lem:oracle4}, no process write to its
$sm.lastOps[i]$ until $rd$ returns, so $sm_i.lastOps[i]=\langle
opNum_i, val \rangle$ when $rd$ returns, and therefore $rd$
returns $val$. Now consider three cases:

\begin{itemize}
  \item There is no \emph{write} operation in $Req_1$ or in
  any request that was chosen before $Req_1$ in $r$. In this
  case, there is no \emph{write} operation before $rd$ in
  $\sigma_r$, and no process writes to
  $sm.value$ before $sm.lastOps[i]$ is set to $\langle opNum_i,
  sm.value\rangle$, and therefore, $rd$ returns $\perp$ as
  expected.
  
  \item There is a \emph{write} operation in $Req_1$ in $r$.
  Consider the \emph{write} operation $w$ in $Req_1$ that is
  invoked by the process with the lowest id, and assume its
  argument is $v'$. Notice that $w$ is the last \emph{write}
  that precedes $rd$ in $\sigma_r$. By the code of the
  \emph{decide} handler, $sm.value$ equals $v'$ at the
  time when $sm.lastOps[i]$ is set to $\langle opNum_i,
  sm.value\rangle$.
  Therefore, $rd$ returns $v'$, which is the value that is
  written by the last \emph{write} operation that precedes it in
  $\sigma_r$.
  
  \item There is no \emph{write} operation in $Req_1$, but there
  is a request that contains a \emph{write} operation and is
  chosen before $Req_1$ in $r$. Consider the last such request
  $Req_2$, and consider the \emph{write} operation $w$ 
  invoked by the process with the lowest id in $Req_2$. Assume
  that $w$'s argument is $v'$, and $Req_2$ was chosen in a
  consensus associated with timestamp $TS'$ (notice that
  $TS'<TS$). By the code of the \emph{decide} handler and
  Lemma \ref{lem:oracle3}, in all the $sm$'s s.t. $sm.ts=TS'+1$,
  the value of $sm.value$ is $v'$. Now, since there
  is no \emph{write} operation in any chosen request between
  $Req_2$ and $Req_1$ in $r$, no process writes to
  $sm.value$ s.t. $TS' < sm.ts < TS$. 
  Hence, when
  $sm.lastOps[i]$ is set to $\langle opNum_i, sm.value\rangle$,
  $sm.value$ equals $v'$, and therefore $rd$ returns $v'$.
  The operation $w$ is the last \emph{write} operation that
  precedes $rd$ in $\sigma_r$. 
  Therefore $rd$ returns
  the value that is written by the last \emph{write} operation 
  that precedes $rd$ in $\sigma_r$. 
     
\end{itemize}

\end{proof}

\begin{corollary}
\label{cor:atomic}

The algorithm of Section \ref{sec:diamondP} is atomic.

\end{corollary}

\subsection{Liveness}

Consider operation $op_i$ invoked at time $t$ by a correct 
process $p_i$ in run of $r$. 
Notice that $r$ is a run with either infinitely or
finitely many invocations. We show that, in both cases, if
$p_i$ is active in $r$, then $op_i$ returns in $r$.

We associate the addition or removal of process $p_j$ by a
process $p_i$ with timestamp that equals
$sm_i.ts$ at the time when the operation returns. The
addition of all processes in $P_0$ is associated with
timestamp 0.

First, we consider runs with infinitely many invocations. In
Lemma \ref{lem:OracleLive1}, we show that for every process $p$,
every $sm$ associated with a larger timestamp than 
$p$'s addition contains $p$
in $sm.cng.mem$. In Observation \ref{obs:infin}, we show that
in a run with infinitely many invocation, for every timestamp
$ts$, there is a completed operation that has a bigger
timestamp than $ts$ at the time of the invocation.
Moreover, after the stabilization time of the FD, 
operations must help all the slow active processes
in order to complete. In Lemma
\ref{lem:OrLiveInfin}, we use the observation to show that any
operation invoked in a run with infinitely many invocations
returns.

Next, we consider runs with finitely many invocations. We show
in Observation \ref{obs:OrLiveMaj} that there is a correct
majority in every up-to-date configuration, and in Lemma
\ref{lem:OrLiveActive}, we show that eventually all the
active members of the last $sm$ adopt it. Then, in
Lemma \ref{lem:OrLiveFin}, we show that every operation
invoked by active process completes.
Finally, in Theorem \ref{theo:OracleLive}, we that the algorithm
satisfies wait-freedom.

\begin{lemma}
\label{lem:OracleLive1} 

Assume the addition of $p_i$ is associated with timestamp
$TS$ in run $r$.
If $p_i$ is active, then $p_i \in sm.cng.mem$ for every
$sm$ s.t.\ $sm.ts \geq TS$.

\end{lemma}

\begin{proof}

The proof is by induction on $sm.ts$. \textbf{Base:} 
If $p_i \in P_0$, then $p_i \in sm.cng.mem$ for all $sm$ s.t.\
$sm.ts=0$.
Otherwise, $\langle add,p_i \rangle$ is part of a request that
is chosen in a consensus associated with timestamp $TS'=TS-1$,
and thus, by with Lemma \ref{lem:oracle3}, $p_i \in
sm.cng.mem$ for all $sm$ s.t.\ $sm.ts=TS'+1$.
\textbf{Induction:} Process $p_i$ is active, so no 
process invokes $\langle remove,p_i
\rangle$, and therefore, together with the validity of 
consensus, no chosen request contains $\langle remove,p_i
\rangle$.
Hence, if $p_i \in sm.cng.mem$ for $sm$ with $sm.ts=k$, then 
$p_i \in sm.cng.mem$ for every $sm$ s.t.\ $sm.ts = k+1$.

\end{proof}

\begin{observation}
\label{obs:infin}

Consider a run $r$ of the algorithm with infinitely many
invocations.
Then for every time $t$ and timestamp $TS$, there is a completed
operation that is invoked after time $t$ by a process with
$sm.ts > TS$ at the time of the invocation.

\end{observation}

\begin{proof}

Recall that $r$ is well-formed. Therefore, there are infinitely
many completed operations in $r$. Now notice that a process
cannot invoke two operations with the same $sm.ts$. 
For every timestamp $TS$, at the time
of the decision in the consensus associated with $TS$, there are
finitely many correct processes. 
All processes whose addition is
associated with timestamps bigger than $TS$, never have 
$sm.ts \leq TS$. Hence, a finite number of
operations are invoked by processes with 
$sm.ts \leq TS$ at the time of the invocation.
And therefore, after every time $t$, there are completed
operations that are invoked by a processes with $sm.ts > TS$
at the time of the invocation.

\begin{comment}

Therefore in every run $r$ with infinitely many completed
operations,
for every timestamp $TS$, there is a finite number of operations
that is invoked by processes s.t\ their $sm.ts$ is equals to
$TS$ at the time of the invocation.

 And thus, in $r$, for every
timestamp $TS'$, there is a completed operation that is invoked
by a process, which $sm.ts$ is equals to $TS$ at the time of the
invocation.

\end{comment}

\end{proof}

\begin{comment}

\begin{observation}
\label{obs:infin}

Notice that a process cannot invoke two operations s.t\ its 
$sm.ts$ is the same at the time of both invocations.
Therefore in every run $r$ with infinitely many completed
operations,
for every timestamp $TS$, there is a finite number of operations
that is invoked by processes s.t\ their $sm.ts$ is equals to
$TS$ at the time of the invocation. And thus, in $r$, for every
timestamp $TS'$, there is a completed operation that is invoked
by a process, which $sm.ts$ is equals to $TS$ at the time of the
invocation.

\end{observation}
\end{comment}

\begin{lemma}
\label{lem:OrLiveInfin}

Consider an operation $op_i$ invoked at time $t$ by an
active process $p_i$ in a
run $r$ with infinitely many invocations. Then $op_i$ completes
in $r$.

\end{lemma}

\begin{proof}

Assume by way of contradiction that $p_i$ is
active and $op_i$ does not complete in $r$.
Assume that $p_i$'s addition is associated with timestamps
$TS$.
Consider a time $t'>t$ after $p_i$ invoked $op_i$ and the FD has
stabilized. By Observation \ref{obs:infin}, there is a completed
operation $op_j$ in $r$, invoked by some process $p_j$
at a time $t'' > t'$ when $sm_j.ts > TS$, which completion is
associated with timestamp $TS'$. By Lemma \ref{lem:OracleLive1},
$p_i \in sm_j.cng.mem$, at time $t''$. 
Now by the algorithm and by the eventual strong
accuracy property of the FD, $p_j$ proposes $op_j$ and $op_i$ in
the same request, and continue to propose both of them until one
is selected. Note that it is impossible for $op_i$ to be
selected without $op_i$ since any process that helps $p_i$ after
stabilization also helps $p_i$.
Hence, since $op_i$ completes, they are both performed in the
same \emph{decide} handler.
The run is well-formed, so $p_i$ does not invoke operations that
are associated with a higher $num$. 
Hence, following the time when $op_i$ is selected, for all
$sm$ s.t.\ $sm.ts > TS'$, $sm.lastOps[i].num = opNum_i$. 
Now,
again by Observation \ref{obs:infin}, consider a completed
operation $op_k$ in $r$, that is invoked by some process $p_k$
at time $t'''$ after the stabilization time of the FD s.t.\
$sm_k.ts > TS'$ at time $t'''$. Operation $op_k$ cannot
complete until $p_i$ receives $p_k$'s $sm$. 
Therefore, $p_i$
receives $sm$ s.t.\ $sm.ts \geq TS'$, and thus
$sm.lastOps[i].num = opNum_i$.
Therefore, $p_i$ learns that $op_i$ was performed, and $op_i$
completes. A contradiction.

\end{proof}

We now proceed to prove liveness in runs with finitely many
invocations.

\begin{definition}

For every run $r$ of the algorithm, and for any point $t$ in
$r$, let $TS_t$ be the timestamp 
associated with the last consensus that made a
decision in $r$ before time $t$. Define $sm^t$, at any point
$t$ in $r$, to be the $sm$'s state after the completion of the
decide handler associated with timestamp $TS_t$ at any process.
Recall that $sm^0$ is the initial state.

\end{definition}

\begin{observation}

For every run $r$ of the algorithm, and for any point $t$ in
$r$, $sm^t$ is unique. 

\end{observation}

\begin{proof}

By Lemma \ref{lem:oracle3}, all the decide handlers associated
with the same timestamp lead to the same $sm$.

\end{proof}

\begin{observation}
\label{obs:OrLiveCond}

For every run $r$ of the algorithm, and for any point $t$ in
$r$, $V(t).members \subseteq 
sm^t.cng.mem$ and $sm^t.cng.mem \cap V(t).remove=\{\}$.

\end{observation}

\begin{proof}

Easy to show by induction.

\end{proof}

\begin{observation}
\label{obs:OrLiveMaj}

For every run $r$ of the algorithm, and for any point $t$ in
$r$.
There is a majority
of $sm^t.cng.mem$ $M$ s.t. $M \subseteq (V(t).members \cup
P(t).join) \setminus F(t)$.

\end{observation}

\begin{proof}

The observation follows from Observation \ref{obs:OrLiveCond}
and the failure condition.

\end{proof}

\begin{observation}

Consider a run $r$ of the algorithm with finitely many
invocations.
Then there is a point $t$ in $r$ s.t.\ for every $t'>t$,
$sm^t=sm^{t'}$. Denote this $sm$ to be $\hat{sm}$. 

\end{observation}

\begin{lemma}
\label{lem:OrLiveActive}

Consider a run $r$ of the algorithm with finitely many
invocations.
Then eventually for
every active process $p_i \in \hat{sm}.cng.mem$,
$sm_i=\hat{sm}$.

\end{lemma}

\begin{proof}

Recall that initially $sm_0.cng.mem=P_0$, and for every process
$p_i \in P_0$, $sm_i=sm_0$. Therefore the lemma holds if
$\hat{sm}=sm_0$. Now assume that $\hat{sm} \neq sm_0$, and
assume that $\hat{sm}.ts=TS$. 
Consider the $sm$ s.t.\ $sm.ts=TS-1$, denote it $sm_{prev}$. By
Lemma \ref{lem:oracle3}, $sm_{prev}$ is well defined. 
Assume
that the decision in the consensus associated with $TS-1$,
denote it $con$, was made at some time $t'$. 
Now consider two
possible cases. 

In first case, there is no \emph{reconfig} operation
that was chosen in $con$ and completes in $r$. 
By Observation
\ref{obs:OrLiveCond}, $V(t').members \subseteq
sm_{prev}.cng.mem$ and $sm_{prev}.cng.mem \cap
V(t').remove=\{\}$ at time $t'$, so by the Observation
\ref{obs:OrLiveMaj} there is a correct majority of
$sm_{prev}.cng.mem$.
By our assumption on consensus, a majority of
$sm_{prev}.cng.mem$ has to propose in order to made a decision.
Therefore, since the majority interact, there is some
active process in $sm_{prev}.cng.mem$ that decides in
$con$, and moves its state to $\hat{sm}$. 
Now recall that the
processes periodically send update messages with their $sm$ to
all the process in their $sm.cng.mem$. 
Therefore, for every
active process $p_i$ in $sm_{prev}.cng.mem$, eventually
$sm_i=\hat{sm}$. 

In the second case, some \emph{reconfig} operation
that was chosen in $con$ completes. Notice that its completion
must be associated with timestamp $TS$. 
Therefore, by
the algorithm (lines
\ref{lineOracle:goto}-\ref{lineOracle:op7} in Algorithm
\ref{algOracle:operation}), a majority of $\hat{sm}.cng.mem$
receives $\hat{sm}$. 
By the failure condition, at least one of these processes
is active. 
Hence, thanks to the periodic update messages,
for every active process $p_i$ in $sm_{prev}.cng.mem$,
eventually $sm_i=\hat{sm}$.

\end{proof}

\begin{lemma}
\label{lem:OrLiveFin}

Consider an operation $op_i$ invoked at time $t$ by an
active process $p_i$ in a run $r$ with finitely many
invocations. Then $op_i$ completes in $r$.

\end{lemma}

\begin{proof}

By Lemma \ref{lem:OracleLive1}, $p_i \in \hat{sm}.cng.mem$,
and by Lemma \ref{lem:OrLiveActive}, there is a point $t'$ in
$r$ s.t.\  $sm_i=\hat{sm}$ for all  $t \geq t'$.
Assume by way of contradiction that $op_i$ does not complete in
$r$. 
Therefore, $op_i$ is either stuck in one of its waits or
continuously iterates in a while loop. In each case, we show a
contradiction.
Denote by $con$ the consensus associated with timestamp
$\hat{sm}.ts$. 
By definition of $\hat{sm}$, no decision is made
in $con$ in $r$.

\begin{itemize}

  \item Operation $op_i$ waits in line \ref{lineOracle:gatheE}
  (Algorithm \ref{algOracle:operation}) forever.
  Notice that $\hat{sm}.cng.rem$ contains all the process that
  were removed in $r$, so, after time $t'$, $p_i$ does not wait
  for a reply from a removed process. 
  By the strong
  completeness property of FD, $p_i$ does not wait for faulty
  processes forever. A contradiction.

  \item Operation $op_i$ remains in the while loop in line
  \ref{lineOracle:op4} (Algorithm \ref{algOracle:operation})
  forever.
  Notice that from time $t'$ till $p_i$ proposes in $con$,
  $pend_i$\emph{=false}. Therefore, $p_i$ proposes in $con$ in
  line \ref{lineOracle:propose} (Algorithm
  \ref{algOracle:operation}), and stays in the while
  after the propose. By Observation \ref{obs:OrLiveMaj}, there
  is a majority $M$ of $\hat{sm}.cng.mem$ s.t.\ $M \subseteq
  V(t).members \cup P(t).join \setminus F(t)$.
  Therefore, by the termination of consensus, eventually a
  decision is made in $con$. A contradiction to the definition
  of $\hat{sm}$.
  
  \item Operation $op_i$ remains in the while loop in line
  \ref{lineOracle:agreeB} (Algorithm \ref{algOracle:operation})
  forever.
  Since it does not remain in the while loop in line
  \ref{lineOracle:op4}, $op_i$ proposes infinitely many
  times, and since each propose is made in a different consensus
  and $p_i$ can propose in a consensus beyond first
  one only once a decision is made in the previous one, 
  infinitely many decisions are made in $r$. A contradiction to
  the definition of $\hat{sm}$.   
  
  \item Operation $op_i$ waits in line \ref{lineOracle:op7}
  (Algorithm \ref{algOracle:operation}) forever.
  Consider two cases. First, $sm_i \neq \hat{sm}$
  when $p_i$ performs line \ref{line:oracleMajRec} (Algorithm
  \ref{algOracle:operation}). In this case, $p_i$ continues
  at time $t'$, when it adopts $\hat{sm}$, because $sm_i.ts>ts$
  hold at time $t'$.
  In the second case ($sm_i = \hat{sm}$ when $p_i$ performs line
  \ref{line:oracleMajRec}), $p_i$ sends \emph{update} message to
  all processes in $\hat{sm}.cng.mem$, and waits for a majority
  to reply.
  By Observation \ref{obs:OrLiveMaj}, there is a
  majority $M$ of $\hat{sm}.cng.mem$ s.t.\ $M \subseteq
  V(t).members \cup P(t).join \setminus F(t)$.
  Therefore, eventually $p_i$ receives replies from all the
  process in $M$, and thus continues. In both cases we have 
  contradiction.
  
\end{itemize}  

Therefore, $p_i$ completes in $r$. 

\end{proof}

\begin{theorem}
\label{theo:OracleLive}

The algorithm of Section \ref{sec:diamondP}
implements a wait-free atomic dynamic storage.

\end{theorem}

\begin{proof}

By Lemmas \ref{lem:OrLiveInfin} and \ref{lem:OrLiveFin}, every
operation, invoked in $r$ by an active process,
completes. And by Corollary \ref{cor:atomic}, the algorithm is
atomic.

\end{proof}

\bibliographystyle{plain}
\bibliography{bibliography}
\end{document}